\documentclass{lmcs}  
\pdfoutput=1

\usepackage{lastpage}
\lmcsdoi{21}{2}{17}
\lmcsheading{}{\pageref{LastPage}}{}{}%
{Sep.~10,~2024}{Jun.~04,~2025}{}


\usepackage{hyperref}
\usepackage{amsmath,amsfonts,amssymb,amscd,amsthm,calc}
\usepackage[utf8]{inputenc}

\newcommand{\darrow}{\!\downarrow}

\newcommand{\la}{\langle}
\newcommand{\ra}{\rangle}
\newcommand{\concat}{\mbox{}^\smallfrown}

\newcommand{\bigset}[1]{\big\{ #1 \big\}}
\newcommand{\restr}{\mbox{\raisebox{.5mm}{$\upharpoonright$}}}

\renewcommand{\leq}{\leqslant}
\renewcommand{\geq}{\geqslant}

\newcommand{\fa}{\forall}

\newcommand{\A}{\mathcal{A}}
\newcommand{\B}{\mathcal{B}}

\newcommand{\G}{\mathcal{G}}
\newcommand{\K}{\mathcal{K}}

\newcommand{\ZFC}{\mathrm{ZFC}}
\newcommand{\GCH}{\mathrm{GCH}}

\newcommand{\dom}{\mathrm{dom}}

\newcommand{\eff}{\mathrm{eff}}

\begin{document}

\title[Completions of Kleene's second model]{Completions of Kleene's second model}


\author[S. A. Terwijn]{Sebastiaan A. Terwijn\lmcsorcid{0000-0002-1464-6908}}

\address{Radboud University Nijmegen, Department of Mathematics\\
P.O. Box 9010, 6500 GL Nijmegen, the Netherlands.}
\email{terwijn@math.ru.nl}  




\begin{abstract}
  \noindent 
We investigate completions of partial combinatory algebras (pcas), 
in particular of Kleene's second model $\K_2$ and generalizations 
thereof.
We consider weak and strong notions of embeddability and completion 
that have been studied before in the literature. 
It is known that every countable pca can be weakly embedded into $\K_2$, 
and we generalize this to arbitrary cardinalities by considering 
generalizations of $\K_2$ for larger cardinals. 
This emphasizes the central role of $\K_2$ in the study of pcas. 
We also show that $\K_2$ and its generalizations have strong completions. 
\end{abstract}

\maketitle

\section{Introduction}

Combinatory algebra supplies us with a large variety of abstract 
models of computation. 
Kleene's second model $\K_2$, first defined in 
Kleene and Vesley~\cite{KleeneVesley}, 
is a partial combinatory algebra defined by an application operator on reals. 
(This in contrast with Kleene's first model $\K_1$, defined in terms of 
application on the natural numbers, which is the setting of classical 
computability theory.) 
Let $\omega^\omega$ denote Baire space, and let $\alpha,\beta\in\omega^\omega$. 
Then the application operator in $\K_2$ can be described by 
\begin{equation}\label{def:K2}
\alpha\cdot\beta = \Phi_{\alpha(0)}^{\alpha\oplus\beta}.
\end{equation}
Here $\Phi_e$ denotes the $e$-th Turing functional, and 
$\alpha\oplus\beta$ the join of $\alpha$ and $\beta$. The application 
is understood to be defined if the right hand side is total.
This definition of application in $\K_2$ is the one from  
Shafer and Terwijn~\cite{ShaferTerwijn}, 
which is not the same as the original coding from Kleene
but essentially equivalent to it, and more user friendly. 
The sense in which the two codings are equivalent is 
explained in Golov and Terwijn~\cite{GolovTerwijn}.
In Section~\ref{sec:K2}, we will also use the original coding of $\K_2$ 
(see Definition~\ref{def:K2k}) when we are discussing larger cardinals, 
for which we do not have a machine model available. 

We will also make use of several variants of $\K_2$.
$\K_2$ has an effective version $\K_2^\eff$, which is obtained by 
simply restricting $\K_2$ to computable elements of $\omega^\omega$. 
The van Oosten model $\B$, first defined in \cite{vanOosten1999}, 
is a variant of $\K_2$ obtained by extending the definition 
of $\K_2$ to partial functions.
This also has an effective version, denoted by $\B^\eff$.

A {\em completion\/} of a pca is an embedding into a 
(total) combinatory algebra. 
Here one can study weak and strong kinds of completion, 
see Definition~\ref{def:embedding}. 
In the strong version, the combinators $k$ and $s$ that every pca
possesses are considered as part of the signature, and in the weak 
version they are not. Both kinds of embedding have been studied in 
the literature. 
The question (posed by Barendregt, Mitschke, and Scott) whether 
every pca has a strong completion was answered in the negative by 
Klop~\cite{Klop}, see also 
Bethke, Klop, and de Vrijer \cite{BethkeKlopVrijer1999}.
On the other hand, $\K_2$ does have strong completions, see 
Theorem~\ref{thm:K2completable} below. 
Weak completions and embeddings were studied for example in 
Bethke \cite{Bethke1988}, 
Asperti and Ciabattoni~\cite{AspertiCiabattoni1997}, 
Shafer and Terwijn~\cite{ShaferTerwijn},  
Zoethout~\cite{Zoethout}, 
Golov and Terwijn~\cite{GolovTerwijn}, and 
more recently in 
Fokina and Terwijn~\cite{FokinaTerwijn}. 
It follows from the main result in Engeler~\cite{Engeler} that every 
pca has a weak completion. Hence fixing the combinators $k$ and $s$ 
or not in embedding results makes a big difference. 
There is an even weaker notion of embeddability of pcas
introduced by Longley, see \cite{Longley} and \cite{LongleyNormann}, 
that is useful in the study of realizability.
We will not consider this notion here, but simply note that the results 
about weak completions obtained here are the strongest possible. 
Extensions of Kleene's second model containing specific functions  
were studied in van Oosten and Voorneveld~\cite{vanOostenVoorneveld}.

It is known that every countable pca is weakly embeddable into $\K_2$. 
This is optimal, since this does not hold for strong embeddings. 
We generalize this to larger cardinalities by considering a 
generalization $\K_2^\kappa$ for larger cardinals $\kappa$ (Section~\ref{sec:K2}). 
We show that $\K_2$ and $\K_2^\kappa$ have strong completions 
(Theorems~\ref{thm:K2completable} and \ref{thm:K2kappa}), 
and prove an embedding result for $\K_2^\kappa$ (Theorem~\ref{thm:emb}).

Our notation for computability theory follows Odifreddi \cite{Odifreddi},
for set theory Kunen \cite{Kunen}, 
and for partial combinatory algebra van Oosten \cite{vanOosten}. 
For unexplained notions we refer to these monographs, 
though we make an effort to make the exposition self-contained by 
listing some preliminaries in Section~\ref{sec:prelim}. 
We use $\la\cdot\ra$ to denote coding of sequences, 
see Section~\ref{sec:K2}.

\section{Preliminaries}\label{sec:prelim}

A {\em partial combinatory algebra\/} (pca) is a set $\A$ together
with a partial map $\cdot$ from $\A\times \A$ to $\A$, that has 
the property that it is combinatorially complete (see below). 
We usually write $ab$ instead of $a\cdot b$, and if this is defined 
we denote this by $ab\darrow$.  
By convention, application associates to the left, 
so $abc$ denotes $(ab)c$.
We always assume that our pcas are nontrivial, i.e.\ have more than 
one element. 
Combinatorial completeness is equivalent to the existence of 
elements $k$ and $s$
with the following properties for all $a,b,c\in\A$:
\begin{align}
ka\darrow &\text{ and } kab = a, \label{k} \\
sab\darrow &\text{ and } sabc \simeq ac(bc) \label{s}.
\end{align}

We have the following notions of embedding for pcas. 
To distinguish applications in different pcas, we also write 
$\A\models a\cdot b\darrow$ if this application is defined in $\A$.

\begin{defi} \label{def:embedding}
For given pcas $\A$ and $\B$, an injection 
$f: \A \to \B$ is a \emph{weak embedding} if for all $a, b \in \A$, 
\begin{equation}\label{emb}
\A\models ab\darrow \; \Longrightarrow \;
\B\models f(a)f(b)\darrow \, = f(ab).
\end{equation}
If $\A$ embeds into $\B$ in this way we write $\A\hookrightarrow \B$.
If in addition to \eqref{emb}, for a specific choice of combinators 
$k$ and $s$ of $\A$, $f(k)$ and $f(s)$ serve as combinators for $\B$, 
we call $f$ a {\em strong embedding}.

A (total) combinatory algebra $\B$ is called a {\em weak completion\/} of $\A$ 
if there exists an embedding $\A \hookrightarrow \B$.
If the embedding is strong, we call $\B$ a {\em strong completion}. 
\end{defi}

There exists  pcas that do not have a strong completion, 
see Klop~\cite{Klop} and 
Bethke, Klop, and de Vrijer \cite{BethkeKlopVrijer1999}.
On the other hand, it follows from Engeler~\cite{Engeler} that 
every pca has a weak completion.\footnote{\label{ftn:Engeler}
The main result in \cite{Engeler} is stated only for structures 
with one binary operator, but on p390 it is mentioned that the 
representation theorem can be generalized to multiple and also to
partial operations, so that in particular it also holds for pcas. 
This claim is further substantiated in the supplementary 
note \cite{Engeler2021}.}
So in particular we see that weak and strong embeddings and 
completions are different. 

We can see that every countable pca has a weak completion as follows. 
In Golov and Terwijn~\cite[Corollary 6.2]{GolovTerwijn} it was shown that 
every countable pca is weakly embeddable into $\K_2$.
Now $\K_2$ is weakly completable, as it weakly embeds, for example, 
into Scott's graph model $\G$ (see Scott~\cite{Scott} for the definition, 
and \cite{GolovTerwijn} for the embedding), 
which is a total combinatory algebra related to enumeration reducibility. 
Hence every countable pcas has a weak completion. Note that this 
includes the counterexample for strong completions from 
Bethke et al.\ \cite{BethkeKlopVrijer1999}, which is countable. 

In the next section we show that 
$\K_2$ and $\K_2^\eff$ have strong completions.
Note that none of these completions is computable, 
as (total) combinatory algebras are never computable.  
The complexity of completions was studied in \cite{Terwijn2023}.

\section{Unique head normal forms} \label{sec:uhnf}

In this section we give two different arguments showing that 
$\K_2$ is strongly completable. 
We start by verifying that $\K_2$ satisfies the criterion from  
Bethke, Klop, and de Vrijer \cite{BethkeKlopVrijer1996} 
about unique head normal forms (hnfs).
This argument that $\K_2$ is completable is longer than the 
short argument at the end of the section, but it is still 
interesting because it yields as a completion a certain 
term model that is potentially different, and it also serves as 
a template for the later generalization to larger cardinalities 
(Theorem~\ref{thm:K2kappa}).

Given a pca $\A$ and a choice of combinators $s$ and $k$, 
the {\em head normal forms\/} (hnfs) of $\A$ are all 
elements of the form $s$, $k$, $ka$, $sa$, and $sab$, 
with $a,b\in \A$. 
Each of these five types of hnf is called {\em dissimilar\/} 
from the other four. 

\begin{defi}[Bethke et al.\ \cite{BethkeKlopVrijer1996}]
$\A$ has {\em unique hnfs\/} if 
\begin{itemize}
\item no two dissimilar hnfs can be equal in $\A$, 
\item {\em Barendregt's axiom\/} holds in $\A$: 
for all $a,a',b,b'\in \A$, 
\begin{equation}\label{BA}
sab=sa'b' \Longrightarrow a=a' \wedge b=b'.
\end{equation}
\end{itemize}
\end{defi}
\noindent 
Bethke et al.\ \cite{BethkeKlopVrijer1996} confirmed a conjecture 
from Klop~\cite{Klop}, namely that every pca with unique hnfs 
is strongly completable. 

\begin{lem}\label{lem:uhfn}
$\K_2$ and $\K_2^\eff$ have unique hnfs. 
\end{lem}
\begin{proof}
As mentioned in the introduction, we will use the coding of $\K_2$
with application defined by 
\[
\alpha\cdot\beta = \Phi_{\alpha(0)}^{\alpha\oplus\beta}
\]
for all $\alpha,\beta\in\omega^\omega$.

We verify that we can choose $k$ and $s$ in $\K_2$ with their 
defining properties \eqref{k} and \eqref{s}, 
and such that all hnfs are different. 
First we note that since $\K_2$ is nontrivial (i.e.\ has more 
than one element) we will automatically have $s\neq k$
(cf.\ Barendregt~\cite[Remark 5.1.11]{Barendregt}).

Somewhat abusing terminology (common in set theory), 
we will refer to the elements of $\omega^\omega$ as reals. 
Without spelling out all syntactic details, we simply note that we 
can define computable reals $k$ and $s$ such that the following hold:
\begin{itemize}

\item $ka = \Phi_{k(0)}^{k\oplus a}$ is a real that stores $a$, 
prepended by a code $ka(0)$ that does not need $a$, so $ka(0)=m$ 
for all~$a$.

\item $kab = \Phi_{ka(0)}^{ka\oplus b}= \Phi_m^{ka\oplus b}$. 
Here $m$ is a code that unpacks $a$ from $ka\oplus b$.

\item $sa = \Phi_{s(0)}^{s\oplus a}$ is a real that stores $a$, 
prepended by a code $sa(0)$ not depending on $a$, so $sa(0)=n$ for all~$a$.

\item $sab = \Phi_{sa(0)}^{sa\oplus b} = \Phi_n^{sa\oplus b}$ is a real 
that stores $a$ and $b$, prepended by a code $sab(0)$ that does not need 
$a$ and $b$, so $sab(0)=p$ for all $a$ and~$b$.

\item $sabc = \Phi_{sab(0)}^{sab\oplus c} = \Phi_p^{sab\oplus c}$. 
Here $p$ is a code that unpacks $a$ and $b$ from $sab\oplus c$, 
and then computes $ac(bc)$.

\end{itemize}
\noindent 
To turn these informal definitions into formal ones, one would first 
give more explicit codings of $ka$, $sa$, and $sab$, and then formally 
define $m,n,p$. However, the details of how one would store these reals 
are not essential, and the above definitions should suffice to convince 
the reader of the existence of $k$ and $s$. Note that $k$ and $s$, 
although elements of $\omega^\omega$, are essentially finite objects 
coding finite programs, hence computable. 
Also, all hnfs defined above are different. 
Note that the definition of $sab$ guarantees Barendregt's axiom \eqref{BA}, 
since $a$ and $b$ can be extracted from $sab$.
This shows that we can choose $k$ and $s$ so that $\K_2$ has 
unique hnfs. 
The argument for $\K_2^\eff$ is the same, since we may use 
the exact same computable reals $k$ and $s$ for $\K_2^\eff$.
So also $\K_2^\eff$ has unique hnfs. 
\end{proof}

\begin{thm}\label{thm:K2completable}
$\K_2$ and $\K_2^\eff$ are strongly completable. 
\end{thm}
\begin{proof}
The theorem follows immediately from Lemma~\ref{lem:uhfn} 
and the criterion provided in Bethke, Klop, and de Vrijer 
\cite{BethkeKlopVrijer1996}.
\end{proof}

We point out again that Theorem~\ref{thm:K2completable} is for 
the coding of $\K_2$ from Shafer and Terwijn~\cite{ShaferTerwijn}. 
This coding is equivalent to the original coding from 
Kleene and Vesley \cite{KleeneVesley} in the sense that 
there are weak embeddings back and forth, as proved in 
Golov and Terwijn~\cite{GolovTerwijn}.
It does not automatically follow from this that Theorem~\ref{thm:K2completable} 
also holds for the original coding. 
That it does will follow from Theorem~\ref{thm:K2kappa}.

We now give a second, much quicker, argument for the strong 
completability of $\K_2$.
Recall that $\B$ denotes van Oosten's model, obtained by extending 
$\K_2$ to partial functions.
The embedding $\K_2\hookrightarrow \B$ is strong, since it is 
just the inclusion, and we can choose the combinators $s$ and $k$ 
of $\K_2$ so that they also work for $\B$. 
Since $\B$ is a total combinatory algebra, 
it follows immediately that $\K_2$ is strongly completable. 
The same remarks apply for the embedding 
$\K_2^\eff\hookrightarrow \B^\eff$ of the effective versions 
of these pcas.

\section{$\K_2$ generalized} \label{sec:K2}

We will make use of the following generalization of $\K_2$, 
taken from van Oosten~\cite[1.4.4]{vanOosten}.

The basic idea underlying $\K_2$ is that every 
partial continuous $f:\omega^\omega \to \omega^\omega$ 
can be coded by a real $\alpha\in\omega^\omega$. 
Namely, $f$ can be defined by its modulus of continuity 
$\hat f:\omega^{<\omega} \to \omega^{<\omega}$, which is 
essentially a function $\alpha:\omega\to\omega$ by 
virtue of $\omega^{<\omega} = \omega$ (as cardinals).
This observation allows use to define a partial application 
$\alpha\cdot \beta$ on $\omega^\omega$ by applying the 
partial continuous function coded by $\alpha$ to~$\beta$.

Now suppose that $\kappa$ is a cardinal such that 
$\kappa^{<\kappa} = \kappa$. 
Then the same ideas apply 
to show that every partial continuous $f:\kappa^\kappa\to\kappa^\kappa$
can be coded by an element $\alpha\in\kappa^\kappa$. 
So again we can define a partial application $\alpha\cdot \beta$ 
on $\kappa^\kappa$ by interpreting $\alpha$ as a 
partial continuous function on $\kappa^\kappa$.
The proof that this yields a pca proceeds along the same 
lines as the proof for $\K_2$. 
This pca is denoted by $\K_2^\kappa$.

For the proofs below (especially that of Theorem~\ref{thm:emb}), 
we need to be more specific about how 
$\alpha\in\kappa^\kappa$ codes a partial continuous function. 
Unfortunately, for large cardinals we have no machine model 
available,\footnote{This could be remedied to some extent by 
using tools from higher computability theory, but also these 
have a limit, and the cardinals we need here may be even larger.}
so that we have to use the less user friendly but more 
concrete coding of Kleene, which we now specify.

Following Kunen~\cite[p34]{Kunen}, we define 
$\kappa^{<\kappa} = \bigcup\{\kappa^\lambda:\lambda<\kappa\}$.
For a function $\sigma\in \kappa^{<\kappa}$, we let 
$\dom(\sigma)$ denote the largest $\lambda<\kappa$ on which 
it is defined. 
By assuming $\kappa^{<\kappa} = \kappa$, we have a bijective 
coding function $\la\cdot\ra: \kappa^{<\kappa} \to \kappa$.
For $n\in \kappa$ and $\sigma$ an element of either $\kappa^{<\kappa}$ 
or $\kappa^\kappa$, we let $n\concat\sigma$ denote the function defined by 
\begin{align*}
(n\concat\sigma)(0) &= n \\
(n\concat\sigma)(m+1) &= \sigma(m) \text{ if } m<\omega \\
(n\concat\sigma)(m) &= \sigma(m) \text{ if } \omega\leq m <\kappa. 
\end{align*}
For $\alpha\in\kappa^\kappa$ and $n\in\kappa$, we let 
$\alpha\restr n$ denote the restriction of $\alpha$ to 
elements smaller than~$n$.

Define a topology on $\kappa^\kappa$ that has as basic open sets
\[
U_\sigma = \bigset{\alpha\in\kappa^\kappa \mid 
\fa n\in\dom(\sigma) \; \alpha(n) = \sigma(n)}
\]
with $\sigma\in \kappa^{<\kappa}$.
Then every $\alpha\in\kappa^\kappa$ defines a partial continuous 
map $F_\alpha:\kappa^\kappa\to \kappa$ as follows: 
\begin{align}
F_\alpha(\beta) = n &\text{ if there exists $m\in\kappa$ such that }
\alpha(\la\beta\restr m\ra) = n+1 \label{app1} \\
&\text{ and } 
\fa k<m \; (\alpha(\la\beta\restr k\ra) =0 \text{ or a limit}).
\label{app2}
\end{align}
If there are no such $n$ and $m$ we let $F_\alpha(\beta)$ be undefined. 

\begin{defi}\label{def:K2k}
$\K_2^\kappa$ is the pca on $\kappa^\kappa$ with application 
$\alpha\cdot\beta$ defined by 
\[
(\alpha\cdot\beta)(n) = F_\alpha(n\concat\beta),
\]
where it is understood that $\alpha\cdot\beta\darrow$ if 
$(\alpha\cdot\beta)(n)\darrow$ for every $n\in\kappa$.
\end{defi}

\begin{lem}\label{lem:repr}
Suppose $F:\kappa^\kappa \times \kappa^\kappa \to \kappa^\kappa$ 
is continuous. Then there exists $\phi\in \kappa^\kappa$ such that 
$\phi\alpha\beta = F(\alpha,\beta)$ 
for all $\alpha, \beta\in \kappa^\kappa$.
\end{lem}
\begin{proof}
This statement can be proved in the same way as  
Lemma~1.4.1 in van Oosten~\cite{vanOosten}. 
The easiest way to do this is by following the proof of 
Lemma~12.2.2 in Longley and Normann~\cite{LongleyNormann}.
\end{proof}
Lemma~\ref{lem:repr} is useful in showing that $\K_2^\kappa$ 
is a pca. The existence of a combinator $k\in\kappa^\kappa$ 
such that $k\alpha\beta = \alpha$ for all 
$\alpha, \beta\in \kappa^\kappa$ follows immediately from the 
lemma. We discuss the existence of the combinator $s$ in 
the proof of Lemma~\ref{lem2:uhfn} below. 
This lemma generalizes Lemma~\ref{lem:uhfn} to $\K_2^\kappa$.
Note that there is no machine model available for $\K_2^\kappa$, 
so that we are forced to use the coding above. 
The resulting proof unfortunately has little beauty in it, 
least of all ``cold and austere'', but it does give us what 
we need, namely Theorem~\ref{thm:K2kappa}.

\begin{lem}\label{lem2:uhfn}
$\K_2^\kappa$ has unique hnfs. 
\end{lem}
\begin{proof}
The existence of a combinator $s\in\kappa^\kappa$ 
such that $s\alpha\beta$ is always defined and 
$s\alpha\beta\gamma \simeq \alpha\gamma(\beta\gamma)$ for all 
$\alpha, \beta, \gamma\in \kappa^\kappa$ follows from 
Lemma~\ref{lem:repr} if we can define a (total) continuous 
binary function $F$ such that 
\begin{equation}\label{want}
F(\alpha,\beta)\gamma \simeq \alpha\gamma(\beta\gamma)
\end{equation}
for all $\alpha, \beta, \gamma$.
Writing out \eqref{want} using Definition~\ref{def:K2k} shows 
exactly how to define $F$, and that it is indeed total and 
continuous in $\alpha$ and $\beta$. 

Now if we can also make $F$ injective, it will follow that $s$
such that $s\alpha\beta = F(\alpha,\beta)$, is also injective, 
and hence satisfies Barendregt's axiom~\eqref{BA}.
However, the definition of $F$ obtained from unwinding \eqref{want} 
does not automatically give that $F$ is injective, so we need to 
do a bit of extra work for this. 
We show that, starting with any $F$ satisfying \eqref{want}, we can 
convert $F$ to a function that is in addition injective. 
The idea is that if $(F(\alpha,\beta)\gamma)(n)$ is defined for just any 
$\gamma\in\kappa^\kappa$ and $n\in\kappa$, then the coding of 
Definition~\ref{def:K2k} leaves ample room to code $\alpha$ and $\beta$ 
in the values of $F(\alpha,\beta)$ that do not matter. 
The only trouble is when there are no defined values at all, in which 
case we have to resort to another solution. 

{\em Case~1.} Suppose $(F(\alpha,\beta)\gamma)(n)$ is defined for some 
$\gamma\in\kappa^\kappa$ and $n\in\kappa$.
By Definition~\ref{def:K2k} this means that 
$F_{F(\alpha,\beta)}(n\concat\gamma)$ is defined. This definition uses 
only an initial segment $\gamma\restr m$ for some $m\in\kappa$,
namely the values 
$F(\alpha,\beta)(\la n\concat \gamma\restr m' \ra)$ with $m'\leq m$.
In particular, 
the higher values 
$F(\alpha,\beta)(\la n\concat\gamma\restr m'\ra)$ for $m'>m$ are irrelevant, 
and we are free to alter them as we like, without affecting the 
effect of $F$.
Now we use these values to code the sequences $\alpha$ and $\beta$. 
Using a bijection $\kappa + \kappa \leftrightarrow \kappa$, we can 
code them as one sequence $\alpha\oplus\beta$ of length $\kappa$, 
by redefining the values 
$F(\alpha,\beta)(\la n\concat\gamma\restr m'\ra)$ with $m'>m$.

{\em Case~2.} $(F(\alpha,\beta)\gamma)(n)$ is undefined for all 
$\gamma\in\kappa^\kappa$ and $n\in\kappa$. In this case we have by 
Definition~\ref{def:K2k} that 
$F(\alpha,\beta)(n)$ is equal to $0$ or a limit for all $n\in\kappa$. 
The effect of $F(\alpha,\beta)$ does not change if we change the value 
from one limit or $0$ to another. W.l.o.g.\ $\kappa>\omega$, since for 
$\kappa = \omega$ we already have a proof (Lemma~\ref{lem:uhfn}).
If $\kappa>\omega$ we have the two values $0,\omega<\kappa$ that we 
can use for coding $\alpha$ and $\beta$. Again we can code these as 
one sequence $\alpha \oplus \beta \in\kappa^\kappa$.
Using a bijection $\kappa^\kappa \leftrightarrow 2^\kappa$ we may 
assume $\alpha\oplus\beta \in 2^\kappa$, and code it using the two values 
$0$ and $\omega$ by 
\[
F(\alpha,\beta)(n) = 
\begin{cases}
0 &\text{if $(\alpha\oplus\beta)(n)=0$}, \\
\omega &\text{if $(\alpha\oplus\beta)(n)=1$}.
\end{cases}
\]

Now since the adapted version of $F$ still has to be continuous, 
we cannot make the case distinction above separately for every 
$\alpha$ and $\beta$: 
For some $\alpha$ and $\beta$ where case~1 applies, 
there may be close $\alpha'$ and $\beta'$ where case~2 applies, 
and since the coding for $\alpha'$ and $\beta'$ starts from the 
beginning, we have no choice but to do the same for 
$\alpha$ and $\beta$.
So we define the adaptation of $F(\alpha,\beta)(\la\sigma\ra))$, 
with $\sigma\in \kappa^{<\kappa}$, by switching between the two cases 
dynamically. As long as we have not found a point yet where case 1 
applies, we follow the coding of case~2, and otherwise we switch to 
the coding of case~1. 
Specifically, if for every $\sigma'\sqsubseteq\sigma$, 
$F(\alpha,\beta)(\la\sigma'\ra))$ is $0$ or a limit, we apply 
the coding of case~2. Otherwise, if $\sigma'\sqsubseteq\sigma$ 
is the minimal initial such that 
$F(\alpha,\beta)(\la\sigma'\ra))$ is a successor, we apply from 
that point onwards the coding of case~1. 

This concludes the proof that we can make $F$ injective without 
changing its behavior. 
As we already noted, the resulting combinator $s$ then satisfies 
Barendregt's axiom~\eqref{BA}.
In particular, all the hnfs $s\alpha\beta$ are unique. 
Since $\alpha$ can be retrieved from $s\alpha\beta$ for any $\beta$, 
it follows that $s\alpha$ is also injective, and hence 
the hnfs $s\alpha$ are also unique. 
It follows from the properties of $k$ that these are also 
different from all $k\alpha$, so that indeed all hnfs are different. 
\end{proof}

\begin{thm}\label{thm:K2kappa}
Every $\K_2^\kappa$ is strongly completable. 
\end{thm}
\begin{proof}
This is again immediate from Lemma~\ref{lem2:uhfn} 
and the criterion provided in Bethke, Klop, and de Vrijer 
\cite{BethkeKlopVrijer1996}.
\end{proof}

For $\K_2$ (i.e.\ for $\kappa=\omega$), we gave two arguments for its
strong completability in Section~\ref{sec:uhnf}. 
Apart from the argument using Lemma~\ref{lem:uhfn} about unique hnfs 
we had an easier option using van Oosten's model~$\B$. 
It would be interesting to generalize this to $\B^\kappa$ for all 
$\kappa$ with $\kappa^{<\kappa} = \kappa$ as above. 
It is not immediately clear how to define $\B^\kappa$. 
Simply restricting $\K_2^\kappa$ to partial functions does not work, 
as this makes use of the machine model for $\K_2$ that is not 
available for $\K_2^\kappa$. Also, the coding from Definition~\ref{def:K2k}
is not suitable for $\B$, as is discussed in 
Longley and Normann~\cite[p77]{LongleyNormann}, because this coding 
does not allow to jump over undefined values. 

The definition of $\K_2^\kappa$ uses the property 
$\kappa^{<\kappa} = \kappa$, 
needed for the coding $\la\cdot\ra: \kappa^{<\kappa} \to \kappa$.
It is known that inaccessible cardinals have this property, 
cf.\ Jech~\cite[p50]{Jech}, 
but the existence of inaccessibles cannot be proved in $\ZFC$
(cf.\ Kunen~\cite{Kunen}).
For $\kappa = 2^\omega$, the property $\kappa^{<\kappa} = \kappa$ 
is independent.\footnote{One can easily check that the statement 
is true under CH and false e.g.\ when 
$2^\omega = \omega_2$ and $2^{\omega_1} = \omega_3$.} 
To guarantee existence of $\K_2^\kappa$ for arbitrary large $\kappa$
we can for example use the generalized continuum hypothesis $\GCH$, 
which says that 
$2^{\omega_\alpha} = \omega_{\alpha+1}$ for every ordinal $\alpha$.
(This is consistent by G\"odels famous result that 
$\GCH$ holds in~$\mathrm{L}$.)
Assuming $\GCH$ we have 
\begin{equation}\label{condition}
\omega_{\alpha+1}^{<\omega_{\alpha+1}} =
(2^{\omega_\alpha})^{\omega_\alpha}= 
2^{\omega_\alpha} = \omega_{\alpha+1}
\end{equation}
so that $\kappa=\omega_{\alpha+1}$ satisfies the hypothesis we need
for $\K_2^\kappa$.

\section{Embeddings into $\K_2^\kappa$}

Having seen that every countable pca weakly embeds into $\K_2$, 
we want to extend this result to larger cardinalities. 
Now we do have the large pcas $\K_2^\kappa$, 
which is a pca on $\kappa^\kappa$,  
but as discussed in Section~\ref{sec:K2} the definition of 
this is contingent on the condition $\kappa^{<\kappa} = \kappa$, 
which implies that for $\kappa>\omega$ at best we know that their existence is 
consistent with $\ZFC$ (using $\GCH$).

\begin{thm}\label{thm:emb}
Assuming $\GCH$, for every pca $\A$ there exists a cardinal $\kappa$
such that $\A$ weakly embeds into $\K_2^\kappa$.
\end{thm}
\begin{proof}
Assume $\GCH$, and let $\A$ be any pca. Let $\alpha$ be an ordinal 
such that $|\A|\leq \omega_{\alpha+1}$, 
and let $\kappa = \omega_{\alpha+1}$.
By \eqref{condition} we have $\kappa^{<\kappa} = \kappa$, and hence 
the pca $\K_2^\kappa$ exists. 
We prove that $\A$ weakly embeds into $\K_2^\kappa$.

The proof from \cite{GolovTerwijn} that every countable pca 
$\A$ embeds into $\K_2$ uses the recursion theorem, 
which we do not have available for large cardinalities, 
so in the proof below we have to avoid this. 
But we can still use the idea that because 
$|\A|\leq \kappa$, the graph of the application operator in 
$\A$ can be coded by an element of $\kappa^\kappa$.

Since $|\A|\leq \kappa$, we can fix a numbering 
$\gamma: \kappa\to\A$. 
If $\gamma(n)=a$ we call $n$ a {\em code\/} of $a$. 
We call $n\in\kappa$ a {\em minimal code\/} if 
for all $m<n$, $\gamma(m)\neq a$.

To obtain an embedding 
$f:\A \hookrightarrow \K_2^\kappa$ we define $f(a)\in\kappa^\kappa$ 
for every $a\in\A$ such that 
\begin{equation}\label{e}
\A\models a\cdot b\darrow = c \;\Longrightarrow \;
\K_2^\kappa \models f(a)\cdot f(b)\darrow = f(c).
\end{equation}

Define a subset $L\subseteq\kappa$ by 
\begin{align*}
L_0 &= \{0\} \\
L_{i+1} &= \{\la m\concat z\ra \mid m\in L_i \wedge z\in\kappa\} \\
L &= \bigcup_{i\in\omega} L_i.
\end{align*}
Note that the definition of $\K_2^\kappa$ depends on the choice of 
the coding of sequences $\la\cdot\ra$. 
W.l.o.g.\ we may assume that in this coding 
$\la\emptyset\ra\neq 0$, and more generally that 
$\la\emptyset\ra \notin L$.\footnote{
Obviously we may {\em choose\/} the bijective coding 
$\la\cdot\ra: \kappa^{<\kappa} \to \kappa$ to satisfy this. 
If the coding was given, and $\la\emptyset\ra = 0$, then the 
definition of $f(a)$ would have to be adapted, and in particular 
we could not use the first bit $f(a)(0)$ to be a code of~$a$, 
which would make the definition of $f(a)$ even uglier.}

\noindent 
Referring to the coding before Definition~\ref{def:K2k}, 
we define $f(a)$ by:
\begin{enumerate}[label={\rm (\roman*)}]

\item {\em Case $L_0$.} $f(a)(0) = w_0 =$  minimal code of $a$.

\item {\em Case $L_1$.} 
$f(a)(\la 0\concat z_0 \ra) =$ minimal code +1 of 
$\gamma(w_0)\cdot\gamma(z_0)\in\A$ if this is defined, and $0$ otherwise. 

\item {\em Case $L_i$.} 
$f(a)(\la n\concat z_0 \ra)$ for $n\in L_{i-1}$, $i\geq 1$, 
is defined as follows. Note that $n$ is of the form 
\begin{equation}\label{n}
n = \la\la\la\ldots\la\la 0\concat z_{i-1}\ra\concat z_{i-2}\ra\concat
\ldots \ra\concat z_1\ra\concat z_0\ra.
\end{equation}
Define 
\begin{align*}
w_1 &\text{ minimal code of } \gamma(w_0)\cdot\gamma(z_0), \\
w_2 &\text{ minimal code of } \gamma(w_1)\cdot\gamma(z_1), \\
&\vdots\\
w_i &\text{ minimal code of } \gamma(w_{i-1})\cdot\gamma(z_{i-1}).
\end{align*}
Define $f(a)(\la n\concat z_0 \ra) = w_i+i$ provided that all 
applications in the definition of $w_1,\ldots,w_i$ are defined, 
and $f(a)(\la n\concat z_0 \ra) =0$ otherwise.

\item $f(a)(n)=0$ for all $n\notin L$.

\end{enumerate}
Note that with this definition, the graph 
$\{(a,b,c)\mid \A\models a\cdot b\darrow=c\}$ is 
coded into $f(a)$ for every~$a$.
Now we should check that \eqref{e} holds.
Suppose that $a\cdot b\darrow=c$ in $\A$. 
We check that Definition~\ref{def:K2k} gives 
$(f(a)\cdot f(b))(n) = f(c)(n)$ for every $n\in\kappa$.
Suppose that $u,v,w$ are minimal codes of $a,b,c$, respectively.

{\em Case $n=0$.}
By \eqref{app1} we have 
$(f(a)\cdot f(b))(0) = 
F_{f(a)}(0\concat f(b)) = 
f(a)(\la 0\concat v\ra)-1 $.
We use here that $f(a)(\la\emptyset\ra)=0$ 
(because we have assumed $\la\emptyset\ra\notin L$)
and $f(a)(\la 0\ra)=0$ (because $\la 0\ra\notin L$), 
so that \eqref{app2} is fulfilled.
Now by (ii) we have $f(a)(\la 0\concat v\ra) =$
minimal code $+1$ of $\gamma(u)\cdot \gamma(v) = c$,
which is $w+1$.
By (i) we have $f(c)(0) = w$, so that indeed 
$(f(a)\cdot f(b))(0) = f(c)(0)$.

{\em Case $n\in L_i$, $i\geq 1$.} 
Suppose $n$ is of the form \eqref{n}. By (iii) we have 
$f(c)(n) = w_i+i$, starting with $w_1$ a minimal code of 
$\gamma(w)\cdot \gamma(z_0)$, where $w$ is now the minimal 
code of~$c$. 

This should equal 
$(f(a)\cdot f(b))(n) = F_{f(a)}(n\concat f(b)) = f(a)(\la n\concat v\ra)-1$.
Indeed, $f(a)(\la n\concat v\ra)$ $=$ $w_i + i +1$, where we start with
$w$ a minimal code of $\gamma(u)\cdot\gamma(v)$, 
$w_1$ a minimal code of $\gamma(w)\cdot\gamma(z_0)$, etc.
So we see that $(f(a)\cdot f(b))(n) = f(c)(n)$, provided that also 
\eqref{app2} is fulfilled. 

For $f(a)(\la n\concat v\ra)$ to produce the desired value, 
we need for \eqref{app2} that both 
$f(a)(\la n\ra)$ and $f(a)(\la\emptyset\ra)$ are~$0$. 
Again, by assumption, $\la\emptyset\ra\notin L$, 
so $f(a)(\la\emptyset\ra)=0$. 
For every $n\in L$ we have $\la n \ra\notin L$, so indeed also
$f(a)(\la n\ra)=0$ by~(iv).
\end{proof}

\section{Concluding remarks on completions}

Asperti and Ciabattoni \cite{AspertiCiabattoni1997} proved that 
if Barendregt's axiom \eqref{BA} holds in a pca then it has 
a weak completion.
(See also the discussion about this in 
Bethke et al.\ \cite[p501]{BethkeKlopVrijer1999}.)
By Engeler~\cite{Engeler} (see footnote~\ref{ftn:Engeler} above), 
the condition about Barendregt's axiom is not needed. 
Note that Barendregt's axiom does not hold for the example 
from Bethke et al.\ \cite{BethkeKlopVrijer1999}
of a pca without strong completions. 
Indeed, on p486 loc.\ cit.\ we have $a$ and $b$ distinct such 
that $s(sk)a=s(sk)b$, which violates Barendregt's axiom. 
But since the example constructed in \cite{BethkeKlopVrijer1999}
is countable, and every countable pca weakly embeds into $\K_2$, 
it also follows from the strong completability of $\K_2$ that it 
has a weak completion. 

Theorem~\ref{thm:emb} shows that, under GCH, every pca 
weakly embeds into a $\K_2^\kappa$ for suitably large $\kappa$. 
Since each $\K_2^\kappa$ is strongly completable, this has as a
corollary that every pca has a weak completion (still assuming 
GCH). The latter statement follows unconditionally from  
Engeler~\cite{Engeler}.
GCH is needed in the proof of Theorem~\ref{thm:emb} to guarantee 
the existence of large cardinals $\kappa$ with 
$\kappa^{<\kappa} = \kappa$.
However, we can get the corollary about completions unconditionally 
by using a binary version of $\K_2^\kappa$ defined on $2^\kappa$ rather 
than $\kappa^\kappa$. 
Namely, for the binary version we need the condition $2^{<\kappa} = \kappa$, 
and ZFC unconditionally proves the existence of arbitrary large $\kappa$ 
with this property.\footnote{Namely, $2^{<\kappa} = \kappa$ holds for 
any $\kappa$ that is a strong limit, i.e.\ such that 
$\lambda<\kappa$ implies $2^\lambda<\kappa$, and the strong limit 
cardinals form a proper class, cf.\ Jech~\cite[p50]{Jech}.} 
This gives an alternative proof of Engeler's result that every pca
is weakly completable. 

Note that the condition $\kappa^{<\kappa} = \kappa$ is needed in the 
definition of $\K_2^\kappa$ as the analog of the coding of finite 
sequences $\omega^{<\omega} = \omega$ that is used in the definition 
of $\K_2$. 
Coding of finite sets also plays a role in the definition of 
Scott's graph model $\G$ mentioned above. 
It is interesting to see how Engeler avoids this condition by 
explicitly putting the finite sets into his definition of a 
graph algebra $G(A)$. This has as advantage that no coding is 
needed. On the negative side, the elements of $G(A)$ are 
somewhat contrived, and much more 
complicated than the elements of $\omega^\omega$ or~$\kappa^\kappa$.

\section*{Acknowledgments}

We thank Henk Barendregt for pointing out that it follows from   
Engeler~\cite{Engeler} that every pca has a weak completion, 
and for sharing the supplementary note 
Engeler~\cite{Engeler2021}.
We also acknowledge the helpful remarks of a referee that 
corrected an error in the proof of Lemma~\ref{lem2:uhfn}.

\bibliographystyle{alphaurl}
\bibliography{completions-lmcs}

\end{document}